\documentclass[pra,twocolumn,preprintnumbers,amsmath,amssymb,superscriptaddress]{revtex4-1}
\usepackage[utf8]{inputenc}
\usepackage{amssymb,amsmath,amsthm}
\usepackage{graphicx} %use graph format
\usepackage{braket}
\usepackage{bm}
\newtheorem{theorem}{Theorem}
\newtheorem{proposition}{Proposition}
\newtheorem{lemma}{Lemma}
\newtheorem{corollary}{Corollary}

\begin{document}

\title{Simple determination of dark states in a general multi-level system}

\author{Kaixuan Zhou}
\affiliation{New York University Shanghai, 1555 Century Ave, Pudong, Shanghai 200122, China}

\author{June Wu}
\affiliation{University of Chicago, Chicago, IL, 60637, USA} 

\author{Junheng Shi}
\affiliation{New York University Shanghai, 1555 Century Ave, Pudong, Shanghai 200122, China} 

\author{Tim Byrnes}
\email{tim.byrnes@nyu.edu}
\affiliation{New York University Shanghai, 1555 Century Ave, Pudong, Shanghai 200122, China}
\affiliation{State Key Laboratory of Precision Spectroscopy, School of Physical and Material Sciences, East China Normal University, Shanghai, 200062, China}
\affiliation{NYU-ECNU Institute of Physics at NYU Shanghai, 3663 Zhongshan Road North, Shanghai, 200062, China}
\affiliation{Center for Quantum and Topological Systems (CQTS), NYUAD Research Institute, New York University Abu Dhabi, UAE}
\affiliation{Department of Physics, New York University, New York, NY 10003, USA}

\date{\today}

\begin{abstract}
In a multi-level energy system with energy transitions, dark states are eigenstates of a Hamiltonian that consist entirely of ground states, with zero amplitude in the excited states.  We present several criteria which allows one to deduce the presence of dark states in a general multi-level system based on the submatrices of the Hamiltonian. The dark states can be shown to be the right-singular vectors of the submatrix that connect the ground states to the excited states.  Furthermore, we show a simple way of finding the dark state involving the determinant of a matrix constructed from the same submatrix.
\end{abstract}

\maketitle

\section{Introduction}
Dark states of atomic systems play an important role in several contexts, ranging from electromagnetically induced transparancy \cite{PhysRevLett.84.5094}, stimulated Raman adiabatic passage (STIRAP) \cite{UNANYAN1998144}, lasing without inversion, slow light \cite{bajcsy2009efficient}, coherent population trapping techniques \cite{donarini2019coherent}, and subrecoil cooling \cite{doi:10.1080/09500340.2019.1646827}.  One of the important effects that takes place in these phenomena is that transitions between atomic levels produce an interference effect such that the amplitude of one of the levels is zero.  The archetypal system that exhibits this effect is the three-level Lambda scheme, where two out of the three transitions between the levels are possible. Such a configuration often arises due to the presence of selection rules, which prevent certain transitions occurring \cite{selection}. A typical situation that occurs is when the intermediate state is an excited state and the other two states are ground states of the atom.  In this case the dark state only involves a superposition of the ground states, and the excited state has zero amplitude \cite{Radmore_1982}.  Such dark states play an important role in STIRAP, where they can be used for achieving an efficient means of population transfer \cite{STIRAP.superconduct,S_rensen_2006}. The use of dark states avoids decoherence via spontaneous emission of intermediate states during the process, due to dark states having precisely zero transition amplitudes.  STIRAP is also resistant to small variations in laser intensity, pulse shape and timing \cite{RevModPhys.89.015006,RevModPhys.70.1003,Bergmann_2019} allowing for robust control of the state transfer.  Numerous experiments have demonstrated the feasibility of STIRAP  \cite{PhysRevLett.99.113003,STIRAP.superconduct,S_rensen_2006,PhysRevLett.122.253201,PhysRevLett.114.205302,PhysRevLett.116.205303}.

%In designing quantum mechanical systems for various applications, particular types of quantum states are more desirable than others for certain tasks \cite{Krantz_2019}. Many complex quantum systems require population transfers between different quantum states \cite{RevModPhys.70.1003}.  STIRAP systems, as well as other , utilize a class of quantum states known as dark states \cite{PhysRevLett.99.113003}. 

%The first STIRAP schemes involved three levels in a $\Lambda$ system. The purpose of these schemes is to transfer a population from one ground state to another without populating the excited state. The process involves optical coupling between the states.

While the three-level Lambda scheme is the most common energy level configuration, several extensions to the basic three-level $\Lambda$ system have been investigated \cite{coulston1992population,doi:10.1063/1.458514}. Most extensions consider a chain model constructed by linking together multiple three-level $\Lambda$ structures \cite{PhysRevA.58.2295,PhysRevA.102.023515}. In these 
multi-level chain models, it is possible to produce dark states for odd numbers of levels \cite{vitanov2001laser}, but dark states are not always possible for an even number states 
%using the standard all-resonant pumping technique
\cite{doi:10.1063/1.461642,PhysRevA.44.R4118,Smith:92}.
Applications of such multi-level dark states include efficient population transfer \cite{Shore:17}, STIRAP-based quantum gates \cite{PhysRevA.65.032318},  entanglement generation \cite{PhysRevA.66.032109,doi:10.1063/1.458514,PhysRevA.98.043616,jing2019split}, and quantum search algorithms \cite{PhysRevLett.99.170503,PhysRevA.78.022330}.  In the general case, one must perform a direct diagonalization of the full Hamiltonian, and obtain the wavefunction of the dark state in this way.

%So we start with the discussion from basic $\Lambda$ models from previous STIRAP research and then the complex multi-level energy models. These discussions on the existence of the dark states in multi-level system proved the importance on the existence of and also the way to look for dark states.

In this paper, we provide a set of results that allows for the simple determination of the existence and wavefunction of a dark state in a multi-level energy network.  The situation we consider is shown in Fig. \ref{Fig1}. We start from a general multi-level energy network where there are transitions between energy levels, potentially involving detuning of the transition energies (Fig. \ref{Fig1}(a)). First a transformation is performed to remove the time dependence of the transitions, such that the energy levels are related to each other by the detunings (Fig. \ref{Fig1}(b)).  We typically assume that the ground states, defined by convention to have zero energy, to have a certain degeneracy, while the excited states have non-zero detunings.  The ground states can be connected to the excited states in an arbitrary fashion, however the ground states do not have any connection to each other.  We then define a dark state as an eigenstate of the Hamiltonian that has exactly zero amplitude with respect to the excited states. We present some general theorems for the existence of a dark state in such a system and show that several simplifications can be made that avoids diagonalization of the full Hamiltonian. We also obtain an explicit expression for the dark state which again does not require diagonalization of the full Hamiltonian.  

We note that there are several related works that have examined similar problems. Moris and Shore (MS) introduced the MS-transformation approach to analyze two-level excitations \cite{PhysRevA.27.906}, and later extended this to more complex systems \cite{doi:10.1080/09500340.2013.837205}. In the most basic form, the MS-transformation reduces a degenerate two energy level system to a set of independent non-degenerate one and two level systems.  
% However, this approach doesn't work so well for those complex cases with states without connections to the ground states. This approach only works if there exists a dark state for the system and it didn't give a specific way to directly see whether the dark states exist or not. 
More recently, Finkelstein-Shapiro, Keller and co-workers introduced the general conditions for the presence of dark state for a general multi-level system where a subset of the states are considered to be the ground states \cite{PhysRevA.99.053829}.  We confirm some of the results found in these works using an alternative approach, and extend them in several ways. One of the extensions that we obtain is the role of excited states that are not directly connected with the ground state.  We show that such states never affect the dark state and can be ignored.  We furthermore obtain an explicit expression for the dark state in terms of the determinant of a matrix that can be constructed from a submatrix of the Hamiltonian.  This allows for a simple analytic determination of the dark states even in a complex network, which can be difficult to obtain due to the computational overheads associate with symbolic computation.   

This paper is organized as follows. In Sec. \ref{sec:2}, we introduce key definitions and describe the problem concretely.  In Sec. \ref{sec:3}, we provide a list of theorems related to dark states in a multi-level energy network.
In Sec. \ref{sec:4}, we apply our theorems to several examples of multi-level networks, and in Sec. \ref{sec:6} we summarize our results.

\begin{figure}[t]
    \includegraphics[width=\columnwidth]{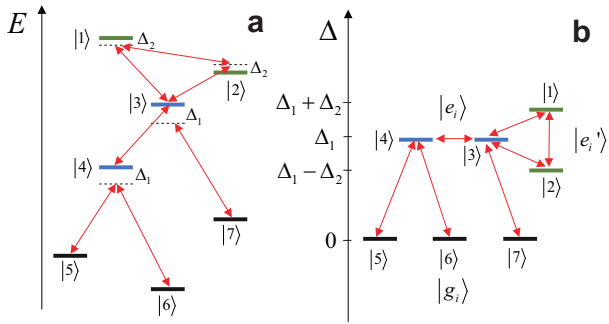}
    \caption{An example multi-level energy network that is considered in this paper. Shown are the (a) physical energy level configuration corresponding to Hamiltonian (\ref{genham}); (b) the equivalent system after performing a unitary transformation to obtain (\ref{detuningham}). Arrows between energy level indicate non-zero transition amplitudes $ \Omega_{ij} $ with frequencies $ \omega_{ij}$ between energy levels $ |i \rangle $ and $ | j \rangle $.  Dashed lines indicate the energy to which the transitions are detuned to, with the associated detunings $ \Delta_1, \Delta_2 $ as marked.  Using the transformation (\ref{unitarytrans}), the time dependent Hamiltonian (\ref{genham}) is transformed to the time independent Hamiltonian (\ref{detuningham}) as shown in (b). }
    \label{Fig1}
\end{figure}

\section{Problem Definition}
\label{sec:2}

Consider a multi-level system consisting of $N$ orthogonal states $|n\rangle$ with $n \in [1,N]$ (see Fig. \ref{Fig1}(a)). The energy of the $n$th energy level is $ E_n $.  Controllable state transitions (e.g. by a laser) are possible between any two states except for connections between two ground states.  Then Hamiltonian of the multi-level network is written
\begin{align}
    {\cal H}= &  \sum_{n} E_n | n \rangle \langle n | \nonumber \\
    & + \hbar \sum_{nm} \left(  e^{-i \omega_{nm} t }  \Omega_{nm} |n \rangle \langle m | + e^{i \omega_{nm} t }  \Omega_{nm}^* |m \rangle \langle n | \right)
    \label{genham}
\end{align}
where the complex amplitude of transition from state $ | m \rangle $ to $ | n \rangle $ is $ \Omega_{nm} $, and $ \omega_{nm} $ is the frequency of the radiation inducing the transition.  We consider two energy levels to be connected if $ | \Omega_{nm} |>0  $. 

We categorize the states into two basic types, either as being a ground state or an excited state (see Fig. \ref{Fig1}(b)).  In principle the choice is arbitrary but in practice the states that are considered to be ground states are low energy states which are stable states in the sense that they do not experience spontaneous emission. The dark states will be defined purely in terms of the ground states.  When labelling only the ground states we use the notation $ | g_i \rangle $, where $ i \in [1,N_g]$.  The remaining $ N - N_g $ states are then excited states.  We further subdivide excited states which are directly connected to the ground states and those which are not.  The $ N_e $ directly connected excited states are labelled as $ | e_i \rangle $, where $ i \in [1,N_e]$.  The remaining disconnected excited state are labelled as $ | e_i' \rangle $, where $ i \in [1,N-N_g-N_e] $.  

Using similar steps to that shown in Ref. \cite{PhysRevA.99.053829}, we may remove the time dependence by performing a unitary transformation on the states 
\begin{align}
   | n \rangle \rightarrow e^{i \epsilon_n t } |n \rangle 
   \label{unitarytrans}
\end{align} 
such that the Hamiltonian becomes
\begin{align}
H  = & \sum_n (E_n - \epsilon_n ) | n \rangle \langle n |  \nonumber \\
& + \hbar \sum_{nm} \left(   \Omega_{nm} |n \rangle \langle m | +  \Omega_{nm}^* |m \rangle \langle n | \right)
\end{align}
where the $ \epsilon_n$ are frequency transformations that satisfy 
\begin{align}
\omega_{nm} = \epsilon_n - \epsilon_m  .
\label{frequencycriterion}
\end{align}
For a $ N $-level system, there are potentially $ N(N-1)/2 $ connections, and only $ N $ variables $ \epsilon_n$. In order for the transformation to be possible in the general case, then the number of transitions must be at most $ N $ such that (\ref{frequencycriterion}) is satisfied.  
%To verify that (\ref{frequencycriterion}) can be satisfied, for each transition in a level structure such as shown in Fig. \ref{Fig1}, cross out one of the energy levels involved in the transition. This amounts to absorbing the transition frequency $ \omega_{nm} $ into the frequency transformation $ \epsilon_n $ for the state $ | n \rangle $.  If each transition can be accounted for in this way, then it is possible to write the Hamiltonian (\ref{genham}) in time-independent form. 
Assuming that the frequencies $ \omega_{nm} $ are controllable, 
the frequency transformations $ \epsilon_n $ allows one to effectively shift the energies of the levels, as long as (\ref{frequencycriterion}) is satisfied.  

One of the results of Ref. \cite{PhysRevA.99.053829} is that 
the ground states must possess a degeneracy in order that a dark state is present.  For this reason we assume that the frequencies $ \epsilon_n $ on the ground states are chosen such that they are all equal
\begin{align}
    E_0 = E_n - \epsilon_n   \hspace{1cm} \text{(ground states)}.  
\end{align}
Redefining the zero level of the energy we may then write
\begin{align}
H   = & \hbar  \sum_{n} \Delta_{n} | n \rangle \langle n |  +  \hbar  \sum_{nm} \left(   \Omega_{nm} |n \rangle \langle m | +  \Omega_{nm}^* |m \rangle \langle n | \right) ,
\label{detuningham}
\end{align}
where the detuning is defined as
\begin{align}
  \hbar \Delta_{n} = E_n - \epsilon_n - E_0 .  
  \label{detunings}
\end{align}
%

% We may rewrite the Hamiltonian in terms of differences of energy levels by rewriting the diagonal part as \cite{PhysRevA.99.053829}
% %
% \begin{align}
% H  = & \frac{1}{2N} \sum_{nm} \Delta_{nm} ( | n \rangle \langle n |  -| m \rangle \langle m | ) \nonumber \\
% & + \hbar \sum_{nm} \left(   \Omega_{nm} |n \rangle \langle m | +  \Omega_{nm}^* |m \rangle \langle n | \right)
% \end{align}
% %
% where the detuning between two energy levels is defined as 
% %
% \begin{align}
%   \Delta_{nm} = E_n - \epsilon_n - E_m +  \epsilon_m  , 
% \end{align}
% %
% and we have dropped an irrelevant energy constant. 
% One of the results of Ref. \cite{PhysRevA.99.053829} is that 
% the ground states must possess a degeneracy in order that a dark state is present.  For this reason we assume that the frequencies $ \epsilon_n $ are chosen such that $ \Delta_{nm} = 0 $ for any pair of $ n,m $ between ground states.  Then using the fact that $ \Delta_{nm} = - \Delta_{mn} $, we may write
% %
% \begin{align}
% H  = & \frac{1}{N} \sum_{nm} \Delta_{nm} | n \rangle \langle n |  \nonumber \\
% & + \hbar \sum_{nm} \left(   \Omega_{nm} |n \rangle \langle m | +  \Omega_{nm}^* |m \rangle \langle n | \right)
% \end{align}

By reordering the Hamiltonian matrix in descending order in energy, such that the last elements are the ground states, we may write the Hamiltonian of such a system in the form
\begin{align}
H =  
\left[
\begin{array}{cc}
A & B \\
B^{\dagger} & 0 
\end{array}
\right] .
\label{chosen_form}
\end{align}
Here, $A$ is the $(N-N_g) \times (N-N_g) $ dimensional submatrix of the Hamiltonian consisting of matrix elements involving only the excited states of the system.
$A$ is a Hermitian matrix containing diagonal elements $A_{nn} = \hbar \Delta_n $. 
%Here we denote the transition frequency between states $|n\rangle$ and $|j\rangle$ as $\Omega_{ij}$. 
According to the Hermitian property of $ A $, the matrix elements $A_{nm} = A_{mn}^* = \hbar \Omega_{nm}$. 
%All transitions between excited states are thus contained in matrix $A$.  
$B$ contains all the transitions between the ground states and the excited states. $B$ is an $(N-N_g) \times N_g$ non-Hermitian matrix with complex elements $B_{nm} = \hbar \Omega_{nm}$. Disconnected energy levels in $A$ and $ B $ have 0 matrix elements.  The dimensions of the square $0$ matrix is $N_g \times N_g$.

%It is convenient to define a symmetric connectivity matrix $J$, where the matrix elements $J_{nm} = J_{mn}=1$ if $|\Omega_{nm}|>0 $, and 0 otherwise.  $\Omega$ has the same zero elements as the $J$ matrix, but the non-zero elements may be any complex number (including zero).  
 
We then define dark states $ |D \rangle $ as an eigenstate of $ H $ that only involve the states in the ground state manifold \cite{PhysRevA.99.053829}.  The dark state then has the form
\begin{align}
|D \rangle  =  
\left[ \begin{array}{cc}
0 \\
|D_g \rangle
\end{array} \right]  , 
\label{darkstate}
\end{align}
%
% where $ |D_g \rangle $ is a $ N_g $ dimensional subvector. Applying (\ref{chosen_form}) to (\ref{darkstate}), we have 
% %
% \begin{align}
% H|D\rangle = \left[ 
% \begin{array}{cc}
% B | D_g\rangle  \\
% 0
% \end{array} \right] .    
% \end{align}
% %
% Demanding that this state is an eigenstate, it then follows that the eigenvalue must be zero
% %
% \begin{align}
%     H \ket{D} = 0 .  
%     \label{zero_eigenvalue}
% \end{align}
% %
We note that this must be true for arbitrary choices of $\Omega_{nm}$.

\section{Properties of the multi-level network}
\label{sec:3}

With the definitions above, we may introduce and prove the main theorems for the presence of dark states in a general multi-level energy network. 
\begin{lemma}
Assuming that $\text{det}(A)\neq 0$, the determinant of the Hamiltonian 
(\ref{chosen_form}) obeys
\end{lemma}
\begin{align}
    \text{det}(H) & =  \text{det}(A)\text{det}(-B^\dagger A^{-1} B) \label{hrelation}  \\
    & = (-1)^N \text{det}(B^\dagger B).
        \label{lemma1}
\end{align}

\begin{proof}
First define an auxiliary matrix
% 
% see https://math.stackexchange.com/questions/1905652/proofs-of-determinants-of-block-matrices
\begin{align}
Q = 
\left[ \begin{array}{cc}
I & -A^{-1}B\\
0 & I\\
\end{array} \right] ,  
\end{align}
where the submatrix dimensions are the same as $ H $.  Multiplying (\ref{chosen_form}) from the left by $ Q $ we have 
\begin{align}
H Q = \left[ \begin{array}{cc}
A & 0\\
B^\dagger & -B^\dagger A^{-1}B\\
\end{array} \right].
\end{align}
Taking the determinant of the left hand side and using the Cauchy-Binet formula, we have $ \text{det}(HQ) = \text{det}(H) $ since $ \text{det}(Q)  = 1 $. Taking the determinants of block matrices \cite{bernstein2009matrix} on the right hand side we obtain (\ref{hrelation}).  
This can be further simplified given the non-singular assumption of $ A $ and the Cauchy-Binet formula again. Since we assume $ \text{det}(A) \neq 0$, we can simplify (\ref{hrelation}) to obtain (\ref{lemma1}).
\end{proof}

We may now relate the dark states of the Hamiltonian (\ref{chosen_form}) to the properties of its submatrix $ B $.
\begin{theorem}
A dark state of the Hamiltonian (\ref{chosen_form}) exists if and only if $B$ has a right-singular vector $|D_g\rangle$ with a singular value of zero.  Furthermore, the dark state has an eigenvalue of zero for Hamiltonian (\ref{chosen_form}).  
\end{theorem}

\begin{proof}
%[TB: could we not just say this already from (3) and (4)?  Or is there a possible loophole in this?]
Performing a singular value decomposition on $ B $ we may write
\begin{align}
    B= \sum_{i} \lambda_i |u_i \rangle \langle v_i |
    \label{svdb}
\end{align}
where $|u_i \rangle, | v_i \rangle  $ are the left- and right-singular vectors and $ \lambda_i$ are the singular values. The singular vectors are orthonormal 
$ \langle u_j | u_i \rangle = \langle v_j | v_i \rangle = \delta_{ij}$.  

First suppose that there exists a singular value such that $ \lambda_j = 0 $, and we take $ |D_g\rangle $ to be the associated right-singular vector $ |D_g\rangle =  | v_j \rangle $. Then from (\ref{svdb}), we have  $ B|D_g\rangle = 0 $, and it immediately follows that 
\begin{align}
    -B^\dagger A^{-1}B|D_g\rangle=0 .
    \label{zeroeigenvalue}
\end{align}
It then follows that that the determinant of the matrix is
\begin{align}
    \text{det}(-B^\dagger A^{-1}B) = 0.
    \label{babyes}
\end{align}
%
%$B$ can be simplified to a square matrix with eigenvectors, which can leads to the result we want. 
Applying Lemma 1, we obtain 
\begin{align}
    \text{det}(H) = 0, 
\end{align}
which shows that a zero eigenvalue state exists. 

We may now explicitly construct the dark state from $ |D_g \rangle $ with the aid of a matrix $P$ of size $N \times N_g $ which projects $|D_g\rangle$ to the whole Hilbert space defined by $H$. 
%Without loss of generality, consider a system with $n$ ground states and $m$ other states. Then B will be an $m*n$ matrix. And $D_g$ has the size $1*n$. 
Using the same submatrix conventions as the right hand column of (\ref{chosen_form}), we define the projector as 
\begin{align}
P = 
\left[ \begin{array}{cc}
 0\\
 I
\end{array} \right] ,   
\end{align}
where the $ 0 $ is a $ (N - N_g) \times N_g $ dimensional zero matrix, and $ I $ is a $ N_g \times N_g $ dimensional identity matrix.  We may use this to relate the full dark state vector $|D\rangle$ to the left singular vector $|D_g\rangle$ via
\begin{align}
    P|D_g\rangle = |D\rangle = \left[ \begin{array}{cc}
0 \\
|D_g \rangle
\end{array} \right] .  
\label{projdarkstate}
\end{align}
Directly applying this state to the Hamiltonian (\ref{chosen_form}) we find that 
\begin{align}
    H | D \rangle = 0 ,
\end{align}
showing that this state is an eigenstate with zero eigenvalue. 
This has the same form as (\ref{darkstate}) showing that a dark state exists.  

Now suppose that a zero singular value $ \lambda_j = 0 $ does not exist.  Applying $ B $ to the state $ |D_g \rangle $ we have
\begin{align}
    B | D_g\rangle = \sum_i \lambda_i \langle v_i | D_g \rangle | u_i \rangle  .  
\end{align}
This cannot be a zero vector as $ |D_g \rangle $ is a normalized state and $ \lambda_j \ne 0 $.  Hence there is no state such that $  B | D_g\rangle = 0 $.  Applying the Hamiltonian to the form of the dark state (\ref{projdarkstate}) for this case, we have
\begin{align}
H |D \rangle = 
\left[ \begin{array}{cc}
B | D_g \rangle  \\
0 
\end{array} \right] .  
\end{align}
For this state to be an eigenstate of the Hamiltonian, we require the eigenvalue of the state to be zero.  However, since here 
\begin{align}
\text{det} (-B^\dagger A^{-1}B) \neq 0
\label{babnone}
\end{align}
using the Cauchy-Binet formula, and $ \text{det} (A)\neq 0$, using  Lemma 1 there is no zero eigenvalue of $ H $.  Hence a dark state does not exist.   

Thus, if and only if there is a right-singular vector $|D_g\rangle$ with a singular value of zero,  a dark state exists in the multi-level network.
\end{proof}

A similar result was found in Ref. \cite{PhysRevA.99.053829} using different arguments.  

\begin{corollary}
A dark state of the Hamiltonian (\ref{chosen_form}) exists if and only if  $ \text{det} (B^\dagger B) = 0 $. 
\end{corollary}
\begin{proof}
From Theorem 1, a dark state exists if and only if $ B|D_g \rangle = 0 $.  Then $ B^\dagger B $ is a square matrix that has a zero eigenvalue, implying  $ \text{det} (B^\dagger B) = 0 $. From Lemma 1, it also implies that $ H $ has a zero eigenvalue.  
\end{proof}
The above result gives a simply way of testing whether a given Hamiltonian possesses a dark state, without computing its eigenvalues directly, and furthermore examining only a submatrix of the whole Hamiltonian.

\begin{proposition}
Transition between excited states, i.e.  without any connections to any ground state, cannot influence the existence of the dark states of multi-level energy network for $\text{det} (A) \neq 0$.
\end{proposition}

\begin{proof}
In the proof for Theorem 1, the only places where the submatrix $ A $ appear are in Eqs. (\ref{babyes}) and (\ref{babnone}), which are statements that are valid for the cases where a right-singular vector such that $ B | D_g \rangle = 0 $ either exists or does not exist, respectively.  Consider changing $ A $ in either of these cases, while maintaining that $\text{det} (A) \neq 0$.  In the case of (\ref{babyes}), the result is unchanged since $ B | D_g \rangle = 0 $.  In the case of (\ref{babnone}), the result is unchanged since the only result that is used here is $\text{det} (A) \neq 0$.  Hence the proof of Theorem 1 is unchanged, and transitions between excited states do not influence the existence of dark states.  
% Recall Lemma 1 and Theorem 1, which states that a dark state $|D\rangle$ can be found by finding the right-singular vector $B|D_g\rangle = 0$ and projecting it into the full Hilbert space of $H$, $P |D_g\rangle = |D\rangle$. Recall also, the product of our auxiliary matrix $Q$ and our Hamiltonian matrix given by
% \begin{align}
%     HQ = \left[ \begin{array}{cc}
% A & 0\\
% B^\dagger & -B^\dagger A^{-1}B\\
% \end{array} \right].
% \end{align}
% By taking the determinant and using the Cauchy-Binet formula we established a necessary condition for the existence that only requires $A$ be non-singular (det$(A) \neq 0$), transitions between excited states are all contained within the sub-matrix $A$
% \begin{align}
%     \Omega_{ij} \geq \forall i,j \leq N-N_g,
% \end{align}
% where Dim$(A) = N- N_g$. Thus, violation of the non-singular requirement, det$(A) \neq 0$, does not directly depend on $\Omega_{ij} \geq \forall i,j \leq N-N_g$ . As the primary requirement (\ref{lemma1}) does not depend on $A$ at all, excited states without connection to ground states
% \begin{align}
%     \Omega_{ij}=0, \hspace{0.5cm} \forall i\leq N-N_g \hspace{0.5cm} and \hspace{0.5cm} j\geq N-N_g,
% \end{align}
% cannot influence the existence of dark states.
\end{proof}

The above result shows that in multi-level networks the excited states that have no direct connections with ground states can be ignored.  As mentioned above, the diagonal elements of submatrix $ A $ contain the detunings of the excited states.  The lack of dependence of the submatrix $ A $ with respect to the presence of dark states means that these detunings can be chosen freely.  For non-zero detunings of the excited states, this makes the assumption $ \text{det}(A) \neq 0$ not very restrictive.  For example, in the absence of excited-to-excited state transitions $ \Omega_{nm} = 0 $, in this case it follows that $ \text{det}(A) \neq 0$.  It is possible to have $ \text{det}(A) = 0$ for particular choices of $ \Omega_{ij} $ where level repulsion of the excited states align one of the energy levels to the ground state level.  However, this will typically require large and carefully chosen values of $ \Omega_{ij} $, and will not give zero energy state for arbitrary $ \Omega_{ij} $ as is required for a dark state.

\begin{proposition}
If the rank of $B$ is less than $N_g$, then a dark state exists.
\end{proposition}

\begin{proof}
Recall $ B $ is a $ N_g \times (N- N_g ) $ matrix.  Hence if $ \text{Rank}(B) = N_g$, it is a full rank matrix and
\begin{align}
\text{det}(B) \neq 0 \implies \text{det}(-B^\dagger A^{-1}B)\neq 0 ,
\end{align}
again assuming $ \text{det} (A) \ne 0$.  Then from Theorem 1, a dark state does not exist. 

If $ \text{Rank}(B) < N_g$, then the number of non-zero singular values is less than the dimension of the space spanned by the ground states $ \{ | g_i \rangle \} $.  Then det$(-B^\dagger A^{-1}B)=0$ and a dark state exists. 
\end{proof}

We note a similar result was found in Ref. \cite{PhysRevA.99.053829}.  The above proposition leads to the following corollary which is useful when designing energy networks to ensure that a dark state is present.  
\begin{corollary}
A network where the number of directly connected excited states $ N_e $ is less than the number of ground state $ N_g$ has at least one dark state.
\end{corollary}
\begin{proof}
Let us write
\begin{align}
B = \left[   \begin{array}{c}
    0 \\
    \tilde{B} 
    \end{array}
    \right]
    \label{bmatzeros}
\end{align}
where $ \tilde{B} $ is a submatrix of $ B $ involving elements with transitions that directly connect between the ground states $ | g_i \rangle $ and the excited states  $ | e_i \rangle $.  By definition, the remaining matrix elements of $ B $ involving disconnected excited states $ | e_i ' \rangle $ are zero.  
A full rank matrix has a rank that is the lesser of the number of rows and columns.  The rank of $ B $ is therefore
\begin{align}
\text{Rank}(B) = \text{Rank}(\tilde{B}) \le N_e.  
\end{align}
For the case that $ N_e < N_g $,  $ \text{Rank}(B) < N_g $.  Then by Proposition 2, a dark state exists. 
\end{proof}

Combining Proposition 1 and Corollary 2, one may easily see that in order to create an energy network that possesses a dark state, it is sufficient to demand that the number of excited states with direct energy connections to ground states is less than the number of ground states. It is clear from this observation that networks with a chain structure such as that given in Fig. 2(a) will always have a dark state, since $ N_e < N_g $.  However, with one additional excited state such that $ N_e = N_g$, the energy configuration will not possess a dark state in general. We note that Corollary 2 is only a sufficient condition for the existence of dark states, and it is possible to have dark states for $ N_e \ge N_g $ but this depends upon the particular connectivity of the energy levels.   

Finally, we discuss procedures of how to obtain the dark state $ | D \rangle $. The most direct way is via Theorem 1, where one finds the singular decomposition of $ B $, and finds the right-singular vector of the zero singular values. We first start with the following Lemma before giving the main result in Proposition 3.  
\begin{lemma}
The dark state $ | D_g \rangle $ is a state that is mutually orthogonal to all the rows of $ B^* $.  
\end{lemma}
\begin{proof}
First define
\begin{align}
    |b_i \rangle = \sum_{j=1}^{N_g} \tilde{B}_{ij}^* | g_j \rangle ,   
    \label{bvecdef}
\end{align}
where $ i \in [1,N_e ] $ and we used the submatrix $\tilde{B}$ is defined in (\ref{bmatzeros}).  
This excludes matrix elements involving excited state that are not directly connected to the ground state.  From Theorem 1, the dark state obeys
\begin{align}
B |D_g \rangle = \sum_{i=1}^{N-N_g} \sum_{j=1}^{N_g} B_{ij} \langle g_j | D_g \rangle | e_i \rangle = 0 .  
\end{align}
We must then have $ \forall i $ 
\begin{align}
\sum_{j=1}^{N_g} B_{ij} \langle g_j | D_g \rangle = 0  .  
\end{align}
Using the definition (\ref{bvecdef}) we see that this is equivalent to 
\begin{align}
    \langle b_i | D_g \rangle = 0 . 
    \label{orthogonalbs}
\end{align}

The remaining rows of $ B^*$ are zero vectors from (\ref{bmatzeros}), hence give a zero inner product with $ |D_g \rangle $. 
\end{proof}

% The above corollary gives an alternative procedure to find a dark state iteratively. [TB: pseudocode table] Starting from a random number, use a Gram-Schmidt orthogonalization procedure to remove the components of $ |b_i \rangle $ from the random vector.  Repeating the procedure, one obtains the dark state $ |D_g \rangle$.  Finally using the relation $ |D \rangle = P | D_g \rangle $ we obtain the dark state in the full Hilbert space.  

We now give an explicit expression for the dark state, which does not require direct diagonalization of the Hamiltonian, or performing a singular value decomposition of $ B $.

\begin{proposition}
Form the set $ \{  | b_i \rangle  \} $,  which are $ N_b $ potentially unnormalized, but linearly independent non-zero vectors, as defined in (\ref{bvecdef}). Suppose that $ N_b < N_g $ such that by Proposition 2, a dark state exists.  Construct the matrix
\begin{align}
{\cal D } = 
\left[ \begin{array}{cccc}
|g_1 \rangle  & | g_2 \rangle  & \dots &  | g_{N_g} \rangle\\
\langle b_1|g_1 \rangle  & \langle b_1| g_2 \rangle  & \dots & \langle b_1| g_{N_g} \rangle \\
\vdots &  & & \\
\langle b_{N_b}|g_1 \rangle &  \langle  b_{N_b} | g_2 \rangle  & \dots & \langle  b_{N_b} | g_{N_g} \rangle \\
\end{array} \right].
\label{caldmat}
\end{align}
Now delete $ N_g - N_b -1 $ columns from $ {\cal D }$ such that the linear independence of the bottom $ N_b $ rows of $ {\cal D } $ is preserved.  Define this square matrix as ${\cal \tilde{D} } $.  Then the dark state is given by 
\begin{align}
    |D \rangle = \det {\cal  \tilde{D} }  . 
    \label{darkstateformula}
\end{align}
\end{proposition}

\begin{proof}
The set of $ \{  | b_i \rangle  \} $ is formed according to the definition (\ref{bvecdef}), but deleting any vectors that are linearly dependent.  Hence $ N_b \le N_e$.  In order to construct the dark state, we require a state that is mutually orthogonal to the vectors $ \{ | b_i \rangle \} $ according to Lemma 2. 
First consider the case $ N_b = N_g - 1$.  In this case the dark state can be found using the generalized cross product \cite{PhysRevA.99.053829,PhysRevA.81.053403} of the matrix $ {\cal  \tilde{D} } = {\cal D } $, since the cross product produces an orthogonal vector to the input vectors.  Using the definition of the generalized cross product (\ref{darkstateformula}) and Lemma 2, we obtain the dark state as given in (\ref{darkstateformula}). 

For cases where  $ N_b < N_g - 1$, the orthogonal vector is not unique since the dimension spanned by $ \{ | b_i \rangle  \} $ is at most $ N_b $. To find a unique orthogonal vector to the spanned space using the generalized cross product, we project out $ N_g - N_b -1 $ of the ground states using the operator
\begin{align}
    P_{b} = I - \sum_{j \in b_{\text{del}} } \frac{| b_j \rangle \langle b_j |}{\langle b_j | b_j \rangle }
\end{align}
where $ b_{\text{del}} $ labels the set of ground states that are deleted.  The vectors in the projected subspace are then 
\begin{align}
   |P( b_i) \rangle = P_b | b_i  \rangle .  
\end{align}
The linear independence property must be maintained to ensure a non-zero result in the generalized cross product.  Forming a matrix similar to (\ref{caldmat}) in the reduced subspace and taking the determinant then gives an orthogonal state to the $ \{ |P( b_i) \rangle \} $.  This state $ |D_g \rangle $ is a state that is orthogonal to the projected vectors $ \langle P(b_i) | D_g \rangle = 0 $ while lying in the space defined by $ P_b $
\begin{align}
    P_b | D_g \rangle = |D_g \rangle.  
\end{align}
As such, it is also orthogonal to the set of $ \{  | b_i \rangle  \} $, and is a dark state according to Lemma 2.  
\end{proof}

Proposition 3 allows one to obtain the dark state without performing a singular value decomposition of $ B $. We can make an extension of this result such that the matrix elements of $ \cal D $ are not necessarily the matrix elements of $ B^*$.  

\begin{corollary}
Construct the matrix ${\cal D }$, defined in (\ref{caldmat}), with the replacement $ |b_i \rangle \rightarrow |b_i' \rangle $, where $ \{ | b_i' \rangle  \} $ a linear combination of $ \{ |b_i \rangle \} $, such that the span of the two sets are the same.  The dark state still takes the form as given in (\ref{darkstateformula}) under this replacement. 
\end{corollary}
\begin{proof}
First we show that using the set of vectors $ \{  | b_i' \rangle  \} $ can be equivalently used to find the dark state in place of $ \{  | b_i \rangle  \} $.  From Lemma 2 we observe the the dark state must be orthogonal to the space spanned by the set of vectors  $ \{  | b_i \rangle  \} $.  A linear combination of such vectors spans the same space as these vectors, hence we may equivalently write the condition of the dark state as 
\begin{align}
    \langle b_i' | D_g \rangle = 0 . 
\end{align}
Then the arguments of Proposition 3 can be repeated with the vectors $ \{  | b_i' \rangle  \} $, which yields the same dark state. 
\end{proof}

The flexibility of using the vectors $ | b_i' \rangle $ instead of $ | b_i \rangle $, is useful when verifying the linear independence property, using methods such as Gaussian elimination.  This will be illustrated in an example in the next section.

\section{Examples}
\label{sec:4}

We now illustrate the use of our results by showing several examples.

\subsection{$\Lambda$ System}

%\begin{figure}[t]
%\includegraphics[width=\linewidth]{Figlambda.eps}
%\caption{This is the simplest case of a multi-level network, consisting of 3 states and 2 transitions. $\Omega$s stands for different frequencies of the transitions, and $|i>$s stands for different states. Notice that all multi-level networks can be reduced to overlapping combinations of different $\lambda$ structures.}
%\label{Fig2.1}
%\end{figure}

Let us first consider the most elementary case of the three-level $\Lambda$ system. The Hamiltonian is 
\begin{align}
H=
\left[ \begin{array}{ccc}
\Delta_1 & \Omega_1 & \Omega_2\\
\Omega_1^* & 0 & 0\\
\Omega_2^* & 0 & 0
\end{array} \right].
\end{align}
The submatrices in the form of (\ref{chosen_form}) are then
\begin{align}
A & = \left[ \Delta_1 \right]  \\
B & =
\left[ \begin{array}{cc}
\Omega_1 & \Omega_2
\end{array} \right] .
\end{align}

To find the dark state, we first construct the vector corresponding to the row of $ B^*$
\begin{align}
|b_1 \rangle = \Omega_1^* |2 \rangle + \Omega_2^* |3 \rangle . 
\end{align}
The matrix (\ref{caldmat}) is then 
\begin{align}
    {\cal D} = 
\left[ \begin{array}{cc}
|2 \rangle & | 3 \rangle \\
\Omega_1 & \Omega_2
\end{array} \right] .
\end{align}
The dark state is then by Proposition 3
\begin{align}
|D\rangle = \text{det} {\cal D}  =  \Omega_2 |2\rangle - \Omega_1 |3\rangle 
\end{align}
where we have omitted the normalization constant.  According to Lemma 2, this state must be orthogonal to $ |b_1 \rangle $, which can be easily verified.  One can obtain the same results by directly calculating the dark state from the singular value decomposition of $B$, where there is one singular value equal to zero.

\subsection{3 layer pyramid}

\begin{figure}[t]
\includegraphics[width=\linewidth]{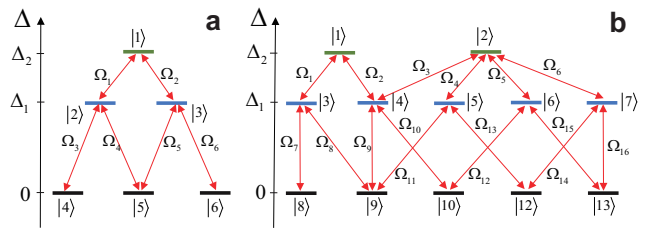}
\caption{Two multi-level energy configurations considered in Sec. \ref{sec:4}.  A (a) $ N = 6  $, $ N_g = 3 $, $ N_e = 2 $; (b) $ N = 13  $, $ N_g = 5 $, $ N_e = 5 $  level system. We use the Hamiltonian after performing the transformation (\ref{unitarytrans}) such that the transition amplitudes are time-inedependent and the associated detunings are marked. }
\label{Fig2.2}
\end{figure}

We next consider the 3 layer pyramid structure as shown in Fig. \ref{Fig2.2}(a).  Writing the Hamiltonian in the form (\ref{chosen_form}), the submatrices are 
\begin{align}
A =
\left[ \begin{array}{ccc}
\Delta_2 & \Omega_1 & \Omega_2 \\
\Omega_1^* & \Delta_1 & 0\\
\Omega_2^* & 0 & \Delta_1\\
\end{array} \right],
\end{align}
and
\begin{align}B = 
\left[ \begin{array}{ccc}
0 & 0 & 0\\
\Omega_3 & \Omega_4 & 0\\
0 & \Omega_5 &\Omega_6\\
\end{array} \right].
\end{align}

Here det$(A) \neq 0$ is satisfied, as long as $ |\Omega_1 |^2 + |\Omega_2 |^2 \ne \Delta_1 \Delta_2$.  Since we aim to have a dark state which has a zero energy for arbitrary values of $ \Omega_i$, we consider this an exceptional case that only gives a zero energy state for specific values of the amplitudes $ \Omega_1 ,\Omega_2$.  

In this example, $ N_e = 2 $ and $ N_g = 3 $, hence Corollary 2 applies and we may already immediately conclude that a dark state exists.  The find the dark state, construct the vectors according to the rows of $ B^*$
\begin{align}
|b_2 \rangle & = \Omega_3^* | 4 \rangle +  \Omega_4^* | 5 \rangle \\
|b_3 \rangle & = \Omega_5^* | 5 \rangle +  \Omega_6^* | 6 \rangle ,  
\label{b2b3}
\end{align}
where we have omitted the normalization factor.  The matrix (\ref{caldmat}) is then 
\begin{align}
    {\cal D} = 
\left[ \begin{array}{ccc}
|4 \rangle & | 5 \rangle & | 6 \rangle \\
\Omega_3 & \Omega_4 & 0 \\
0 & \Omega_5 &\Omega_6 \\
\end{array} \right] .
\end{align}
The dark state is thus
\begin{align}
|D\rangle = \text{det} {\cal D}  = \Omega_4 \Omega_6|4\rangle - \Omega_3 \Omega_6 |5\rangle + \Omega_3 \Omega_5 |6\rangle ,
\end{align}
where again we have omitted the normalization constant. As implied by Proposition 1, there is no dependence of the dark state on any of the parameters that appear in submatrix $ A $. This state is orthogonal to (\ref{b2b3}) as stated in Lemma 2.

%In order to calculate det$(-B^\dagger A^{-1}B)$, the rank of $B$ must be greater than the number of the ground states. For this example, without any restrictions on energy levels and laser frequencies, $\text{Rank}(B) < N_g$ is satisfied in all cases.

\subsection{12 level system}

The two previous examples are well-known energy network structures where it is known that a dark state is present.  We now consider a more non-trivial 12 level example as shown in Fig. \ref{Fig2.2}(b).  The submatrices of the Hamiltonian in the form (\ref{chosen_form}) are
\begin{align}
A=
\left[ \begin{array}{ccccccc}
\Delta_2 & 0 & \Omega_1 & \Omega_2 & 0 & 0 & 0\\
0 & \Delta_2 & 0  & \Omega_3 & \Omega_4 & \Omega_5 & \Omega_6\\
\Omega_1^* & 0 & \Delta_1 & 0 & 0 & 0 & 0\\
\Omega_2^* & \Omega_3^* & 0  & \Delta_1 & 0 & 0 & 0\\
0 & \Omega_4^* & 0 & 0 & \Delta_1 & 0 & 0\\
0 & \Omega_5^* & 0 & 0 & 0 & \Delta_1 & 0\\
0 & \Omega_6^* & 0 & 0 & 0 & 0 & \Delta_1
\end{array} \right] 
\end{align}
and
\begin{align}
B=
\left[ \begin{array}{ccccc}
0 & 0 & 0 & 0 & 0\\
0 & 0 & 0 & 0 & 0\\
\Omega_7 & \Omega_8 & 0 & 0 & 0\\
0 & \Omega_9 & \Omega_{10} & 0 & 0\\
0 & \Omega_{11} & 0 & \Omega_{12} & 0\\
0 & 0 & \Omega_{13} & 0 & \Omega_{14}\\
0 & 0 & 0 & \Omega_{15} & \Omega_{16}
\end{array} \right] 
\end{align}
While numerically diagonalizing a $ 12 \times 12 $ Hamiltonian is not a difficult task, obtaining an analytic expression for the dark states is not as easy.  For example, we were unable to diagonalize the Hamiltonian (\ref{chosen_form}) for $ A, B$  as above using a symbolic computational program such as Mathematica, to show that a dark state exists.  

For this particular example, $ N_g = N_e$, hence we cannot immediately apply Corollary 2, which ensures the existence of a dark state.  Let us first attempt to apply Proposition 2, which requires evaluation of the rank of $ B $. 
Performing Gaussian elimination on $ B $, we obtain
\begin{align}
\left[ \begin{array}{ccccc}
\Omega_7 & \Omega_8 & 0 & 0 & 0\\
0 & \Omega_9 & \Omega_{10} & 0 & 0\\
0 & 0 & \Omega_{13} & 0 & \Omega_{14}\\
0 & 0 & 0 & \Omega_{15} & \Omega_{16} \\
0 & 0 & 0 & 0 & \chi  \\
0 & 0 & 0 & 0 & 0 \\
0 & 0 & 0 & 0 & 0 
\end{array} \right] 
\label{gausselimb}
\end{align}
where
\begin{align}
\chi =  \Omega_9 \Omega_{12} \Omega_{13} \Omega_{16} - \Omega_{10} \Omega_{11} \Omega_{14} \Omega_{15} .  
\end{align}
Hence in general  $ \text{Rank}(B)= N_g  = 5  $, unless 
\begin{align}
\Omega_9 \Omega_{12} \Omega_{13} \Omega_{16} = \Omega_{10} \Omega_{11} \Omega_{14} \Omega_{15}
\label{omegadeps}
\end{align}
where the $ \text{Rank}(B) = 4 $.  This equation can be satisfied by introducing the dependence of one of the transitions implied by (\ref{omegadeps}).  Alternatively, a physically simpler procedure involves setting one of the transitions in $ \{ \Omega_9,  \Omega_{12}, \Omega_{13}, \Omega_{16} \} $  and another one of the transitions in $ \{\Omega_{10}, \Omega_{11} \Omega_{14}, \Omega_{15} \} $ to zero.  Either way (\ref{omegadeps}) is satisfied and by Proposition 2, a dark state exists.   

%This form highlights the two different paths from the state $|8\rangle$ to $|12\rangle$, ignoring the common lasers $\Omega_7$ and $\Omega_8$. With this equation, despite the equal number of ground and first excited states, we still end up satisfying our Proposition 2 condition as $\text{Rank}(B) = 4$, and $N_g = 5$. Allowing us to ignore the states above the $\Delta_1$ level according to Proposition 1. 

Now we compute explicitly the dark states. We use the result of Corollary 3, which states that we may use any equivalent set of linearly independent states $ |b_i' \rangle $, not necessarily the vectors $ |b_i \rangle $ which are formed from the rows of $ B^*$.  We thus use (\ref{gausselimb}) to form the matrix
\begin{align}
{\cal D} = 
\left[ \begin{array}{ccccc}
|8 \rangle  & |9 \rangle & |10 \rangle & |11 \rangle & |12 \rangle \\
\Omega_7 & \Omega_8 & 0 & 0 & 0\\
0 & \Omega_9 & \Omega_{10} & 0 & 0\\
0 & 0 & \Omega_{13} & 0 & \Omega_{14}\\
0 & 0 & 0 & \Omega_{15} & \Omega_{16} 
\end{array} \right] 
\end{align}
Then according to Proposition 3, the dark state is 
\begin{align}
|D \rangle = & \text{det} {\cal D } \nonumber \\
  = &  -\Omega_8  \Omega_{10}  \Omega_{14}  \Omega_{15} | 8 \rangle
+ \Omega_7  \Omega_{10}  \Omega_{14}  \Omega_{15} | 9 \rangle  \nonumber \\
 & -\Omega_7  \Omega_9  \Omega_{14}  \Omega_{15} | 10 \rangle 
-\Omega_7  \Omega_9  \Omega_{13}  \Omega_{16} | 11 \rangle \nonumber \\
 & + \Omega_7 \Omega_{9} \Omega_{13}  \Omega_{15} | 12 \rangle 
\end{align}
omitting normalization factors.  This example shows that one may start with a system that originally does not contain a dark state, but additional conditions may be deduced and applied such that the presence of a dark state is ensured.

%This model is not a regular chain-wise connected model and contains equal numbers of ground states and $\Delta_1$ states. By applying Proposition 2 with an additional condition we are still able to find dark states without full Hamiltonian eigendecomposition. This demonstrates the strength of our approach in handling complex multi-level network without restrictions.

%Notice that in order to simplify our result for the dark states, we didn't take the lasers from level above $\Delta_1$ to the ground level. But all the theorem still works for this structure once the $rank (B)<N$ as Proposition 2 stated.

\section{Summary and conclusions}
\label{sec:6}

We have presented several results to determine the existence of a dark state in a multi-level energy system.  The key result is given in Theorem 1, which states that the dark states are zero eigenvalue states of the submatrix $ B $, as defined in the Hamiltonian (\ref{chosen_form}).  From this basic result, several other key results were deduced, such as the independence on transitions not involving the ground states (Proposition 1) and the relationship between the presence of a dark state and the rank of $ B $ (Proposition 2).  This also allowed us to obtain an explicit expression for the dark state in terms of the determinant of a matrix constructed using the elements of $ B $.  Our set of results allows for the simple determination of whether a dark state is present in an energy network of the general form as given in Fig. \ref{Fig1}(b).  The utility of our methods is that one does not need to explicitly diagonalize the full Hamiltonian in order to determine whether a dark state exists.  According to Corollary 1, finding the determinant of $ B^\dagger B $ allows one to detect whether a dark state is present.  Furthermore, using Proposition 3, one may explicitly calculate the dark state, again using the determinant of a matrix calculated from $ B $.  

While we did not include effects of dissipation is our analysis, the results of Ref. \cite{PhysRevA.99.053829} would suggest that these are unimportant so long as only relaxation effects towards the ground state are considered. The reason is that this type of relaxation by definition has no effect on dark states, which consist entirely of ground states.  In some applications of dark states, such as electromagnetically induced transparency such an assumption may not necessarily be true and may be scope for further work.  

We expect our results to be useful in obtaining analytic expressions for the dark states.  In several applications of dark states it is useful to know the analytic form of a dark state.  While numerical diagonalization can be performed for high dimensional Hamiltonians, obtaining analytical expressions can be much more difficult, using symbolic algorithms (e.g. Mathematica).  Using our methods, one can find the the existence and determine  a dark state simply by evaluating the determinant of a matrix, which can be done efficiently.  In this way, large scale multi-level networks can also be handled.  Such systems may occur in bosonic systems where the indistinguishable nature of the particles produce naturally large Hilbert spaces, that do not always reduce to effectively single particles \cite{PhysRevA.98.043616,abdelrahman2014coherent,byrnes2021quantum}.

\section{Acknowledgements}

This work is supported by the National Natural Science Foundation of China (62071301); NYU-ECNU Institute of Physics at NYU Shanghai; the Joint Physics Research Institute Challenge Grant; the Science and Technology Commission of Shanghai Municipality (19XD1423000,22ZR1444600); the NYU Shanghai Boost Fund; the China Foreign Experts Program (G2021013002L); the NYU Shanghai Major-Grants Seed Fund.

\bibliography{ref}

\end{document}